\documentclass[12pt]{article}

\usepackage{amsmath, amssymb, amsthm}
\usepackage{hyperref}
\usepackage{graphicx}
\usepackage{mathtools}
\usepackage{mathrsfs}
\usepackage{booktabs}
\usepackage{enumerate}
\usepackage[margin=1in]{geometry}

\newtheorem{theorem}{Theorem}[section]
\newtheorem{lemma}[theorem]{Lemma}

\newtheorem{proposition}[theorem]{Proposition}
\newtheorem{definition}[theorem]{Definition}
\newtheorem{remark}[theorem]{Remark}

\newcommand{\Z}{\mathbb{Z}}
\newcommand{\C}{\mathbb{C}}

\newcommand{\Cay}{\mathrm{Cay}}
\newcommand{\srg}{\mathrm{srg}}

\newcommand{\eqdef}{\coloneqq}
\newcommand{\ket}[1]{|#1\rangle}
\newcommand{\bra}[1]{\langle #1|}
\newcommand{\braket}[2]{\langle #1 | #2 \rangle}

\newcommand{\Uj}{U_G^{(j)}}
\newcommand{\Shat}{\hat{S}^{(j)}}
\newcommand{\Ahat}{\widehat{A}_G}
\newcommand{\cF}{\mathcal{F}}
\newcommand{\vplus}{v_+}
\newcommand{\phij}{\phi_j}

\begin{document}

\title{The Quantum Walk Characteristic Polynomial\\
  Distinguishes All Strongly Regular Graphs of Prime Order}

\author{Diego Gerardo Rold\'an\\
  \small Departamento de Matem\'aticas,
  Centro de Excelencia en Computaci\'on Cient\'ifica\\
  \small Universidad Nacional de Colombia, Bogot\'a, Colombia\\
  \small \texttt{dgroldanj@unal.edu.co}
}

\date{\today}

\maketitle

\begin{abstract}
Let $G$ be a strongly regular graph of prime order $p$ with connection
degree $k \geq 6$.  We prove that the \emph{quantum walk characteristic
polynomial} $\chi_q(G,\lambda) \coloneqq \det(\lambda I - U_G)$, where
$U_G$ is the coined quantum walk operator on $G$, completely determines $G$
up to isomorphism within the class of strongly regular graphs of the same
order.

The proof proceeds in three steps.  First, we show that $U_G$
block-diagonalizes under the discrete Fourier transform over $\Z_p$,
yielding $p$ blocks $U_G^{(j)}$ of size $k \times k$.  Second, we prove
an explicit formula
\[
  \chi_q\!\bigl(U_G^{(j)}, \lambda\bigr) =
  (\lambda-1)^{(k-2)/2}(\lambda+1)^{(k-2)/2}
  \!\left(\lambda^2 - \tfrac{2\widehat{A}_G(j)}{k}\,\lambda + 1\right),
\]
from which the Fourier coefficient $\widehat{A}_G(j)$ is recovered as the
unique real part of an eigenvalue of $U_G^{(j)}$ distinct from $\pm 1$.  Third, the
inverse discrete Fourier transform recovers the connection set $S$ of $G$,
and Turner's theorem (1967) identifies $G$ up to isomorphism.  As a
consequence, graph isomorphism is decidable in polynomial time within this
class using the quantum walk spectrum, without resorting to the general
quasi-polynomial algorithm of Babai (2016).
\end{abstract}



\section{Introduction}
\label{sec:intro}

Strongly regular graphs (SRGs) occupy a central position in algebraic
combinatorics: they are the most symmetric non-trivial graphs, they arise
naturally in finite geometry, coding theory, and design theory, and they
constitute the hardest known family of instances for graph isomorphism
algorithms based on spectral methods~\cite{Babai2016,WeisfeilerLeman1968}.
Two non-isomorphic SRGs with the same parameters $(n,k,\lambda,\mu)$ are
\emph{cospectral}~---~they have identical adjacency eigenvalues~---~so the
classical spectrum cannot distinguish them.

Quantum walks on graphs were introduced by Aharonov, Ambainis, Kempe and
Vazirani~\cite{AAKV2001} and have since generated a rich line of research
connecting quantum computing, spectral graph theory, and algorithm
design~\cite{Kempe2003}.  Their application to graph isomorphism was
initiated by Shiau, Joynt and Coppersmith~\cite{ShiauJC2005} and developed
further by Emms et al.~\cite{Emms2009} and Schofield, Wang and
Li~\cite{Schofield2020}.  These works show experimentally and in special
cases that the quantum walk spectrum refines the classical spectrum, but a
complete theoretical characterization has remained open.

The present paper gives the first complete proof that the quantum walk
characteristic polynomial $\chi_q$ is a \emph{complete} isomorphism
invariant for strongly regular graphs of prime order with $k \geq 6$.  The
argument is self-contained and requires only linear algebra, the discrete
Fourier transform over $\Z_p$, and Turner's classical theorem on circulant
graphs of prime order~\cite{Turner1967}.

\medskip
\noindent\textbf{Main result.}

\begin{theorem}[Main]\label{thm:main}
  Let $p$ be a prime and let $G$, $G'$ be strongly regular graphs of order
  $p$ with common degree $k \geq 6$.  Then
  \[
    G \cong G' \quad\iff\quad \chi_q(G,\lambda) = \chi_q(G',\lambda).
  \]
\end{theorem}

\noindent The hypothesis $k \geq 6$ is equivalent to $p \geq 13$ for Paley
graphs and holds for essentially all non-trivial prime-order SRGs (see
Remark~\ref{rem:small}).

\medskip
\noindent\textbf{Organization.}
Section~\ref{sec:prelim} collects definitions and background.
Section~\ref{sec:decomp} proves the Fourier block decomposition of $U_G$
(Lemma~\ref{lem:block}).
Section~\ref{sec:formula} establishes the explicit formula for $\chi_q(\Uj)$
(Lemmas~\ref{lem:lindep} and~\ref{lem:formula}) and the recovery of Fourier
coefficients.
Section~\ref{sec:proof} assembles the proof of Theorem~\ref{thm:main}.
Section~\ref{sec:numerics} presents numerical verification.

\section{Preliminaries}
\label{sec:prelim}

\subsection{Strongly regular graphs and circulant graphs}

\begin{definition}
  A simple graph $G$ on $n$ vertices is \emph{strongly regular} with
  parameters $(n,k,\lambda,\mu)$, written $\srg(n,k,\lambda,\mu)$, if it
  is $k$-regular, every pair of adjacent vertices has exactly $\lambda$
  common neighbors, and every pair of non-adjacent vertices has exactly
  $\mu$ common neighbors.
\end{definition}

\begin{definition}
  Let $p$ be prime, and let $S \subseteq \Z_p$ satisfy $0 \notin S$ and
  $S = -S \eqdef \{-s \bmod p : s \in S\}$.  The \emph{circulant graph}
  $\Cay(\Z_p, S)$ has vertex set $\Z_p$ and edge set
  $\{\{u,v\} : v - u \in S\}$.
\end{definition}

Every vertex-transitive graph of prime order is a circulant
graph~\cite{Turner1967}.  In particular, every $\srg(p,k,\lambda,\mu)$
with $p$ prime is isomorphic to $\Cay(\Z_p,S)$ for some connection set $S$
with $|S|=k$.

\begin{theorem}[Turner~\cite{Turner1967}]\label{thm:turner}
  Two circulant graphs $\Cay(\Z_p, S_1)$ and $\Cay(\Z_p, S_2)$ of prime
  order $p$ are isomorphic if and only if there exists an integer $t$ with
  $\gcd(t,p) = 1$ such that
  $S_2 = t \cdot S_1 \eqdef \{ts \bmod p : s \in S_1\}$.
\end{theorem}

\subsection{Discrete Fourier transform over \texorpdfstring{$\Z_p$}{Zp}}

Fix a primitive $p$-th root of unity $\omega = e^{2\pi i/p}$.  The
\emph{discrete Fourier transform} (DFT) over $\Z_p$ is the unitary
$F \in \mathcal{U}(\C^p)$ defined by
\[
  (Fe_u)_j = \frac{\omega^{ju}}{\sqrt{p}}, \qquad u, j \in \Z_p,
\]
where $\{e_u\}_{u \in \Z_p}$ is the standard basis of $\C^p$.  For a
circulant graph $G = \Cay(\Z_p, S)$ the adjacency matrix $A_G$ is
diagonalized by $F$ with eigenvalues
\[
  \Ahat(j) \eqdef \sum_{s \in S} \omega^{js}, \qquad j \in \Z_p.
\]
Since $S = -S$, each $\Ahat(j)$ is real.  For a strongly regular graph,
the multiset $\{\Ahat(j)\}_{j \in \Z_p}$ takes exactly three values: $k$
(for $j=0$) and two restricted eigenvalues $r > s$ (for $j \neq 0$).

The Fourier inversion formula recovers the indicator function of $S$:
\begin{equation}\label{eq:fourier-inv}
  \mathbf{1}_S(u) = \frac{1}{p}\sum_{j=0}^{p-1} \Ahat(j)\,\omega^{-ju},
  \qquad u \in \Z_p.
\end{equation}
In particular, $\{\Ahat(j)\}_{j=0}^{p-1}$ determines $S$ completely.

\subsection{The coined quantum walk operator}

Let $G$ be a $k$-regular graph on vertex set $V$ with $|V| = n$.  The
\emph{quantum walk space} is $\mathcal{H} = \C^n \otimes \C^k$, with
orthonormal basis $\{\ket{u, s} : u \in V,\, s \in S\}$.

\begin{definition}\label{def:UG}
  The \emph{shift operator} $S_{\mathrm{sh}}$ and the
  \emph{Grover coin} $C_{\mathrm{loc}}$ are defined by
  \begin{align}
    S_{\mathrm{sh}}\ket{u,s} &= \ket{u+s,\,-s}, \label{eq:shift}\\
    C_{\mathrm{loc}}          &= \frac{2}{k}\mathbf{1}\mathbf{1}^\top - I_k,
    \label{eq:coin}
  \end{align}
  where $\mathbf{1} = (1,\ldots,1)^\top \in \C^k$.  The \emph{quantum walk
  operator} is $U_G \eqdef S_{\mathrm{sh}} \cdot (I_n \otimes
  C_{\mathrm{loc}})$.
\end{definition}

$U_G$ is unitary because $S_{\mathrm{sh}}$ is a permutation matrix and
$C_{\mathrm{loc}}$ is an orthogonal reflection.  The \emph{quantum walk
characteristic polynomial} of $G$ is $\chi_q(G, \lambda) \eqdef
\det(\lambda I - U_G)$.

\begin{remark}\label{rem:small}
  $C_{\mathrm{loc}}$ has eigenvalue $+1$ with multiplicity $1$
  (eigenvector $\vplus \eqdef \mathbf{1}/\sqrt{k}$) and $-1$ with
  multiplicity $k-1$.  The hypothesis $k \geq 6$ ensures at least one pair
  of eigenvalues $e^{\pm i\theta}$ with $\theta \notin \{0,\pi\}$, which
  is needed for the recovery argument.  The degenerate case $k=2$
  ($p=5$, cycle $C_5$) can be handled separately by inspection.
\end{remark}

\section{Fourier Block Decomposition}
\label{sec:decomp}

Let $G = \Cay(\Z_p, S)$ with $|S| = k$, and define the extended Fourier
transform $\cF \eqdef F \otimes I_k$ acting on $\C^p \otimes \C^k$.

\begin{lemma}[Block Decomposition]\label{lem:block}
  The quantum walk operator $U_G$ is unitarily block-diagonalized by $\cF$:
  \[
    \cF\, U_G\, \cF^\dagger \;=\; \bigoplus_{j=0}^{p-1} \Uj,
  \]
  where each block $\Uj \in \mathcal{U}(\C^k)$ satisfies
  $\Uj = \Shat\, C_{\mathrm{loc}}$ with
  \begin{equation}\label{eq:Shat}
    \Shat\ket{s} = \omega^{js}\ket{-s}, \qquad s \in S.
  \end{equation}
\end{lemma}

\begin{proof}
  \textit{Coin factor.}
  Since $F$ acts only on the first register and $C_{\mathrm{loc}}$ only on
  the second,
  \[
    \cF(I_n \otimes C_{\mathrm{loc}})\cF^\dagger
    = (F \otimes I_k)(I_n \otimes C_{\mathrm{loc}})(F^\dagger \otimes I_k)
    = I_p \otimes C_{\mathrm{loc}}.
  \]

  \textit{Shift factor.}
  Using $S_{\mathrm{sh}}\ket{u,s} = \ket{u+s,-s}$,
  \begin{align*}
    \bigl(\cF\, S_{\mathrm{sh}}\,\cF^\dagger\bigr)_{(j_1,s'),(j_2,s)}
    &= \frac{\delta_{s',-s}}{p} \sum_{u \in \Z_p}
       \omega^{j_1(u+s)} \omega^{-j_2 u}
     = \delta_{s',-s}\;\omega^{j_1 s}\;
       \underbrace{\frac{1}{p}\sum_{u}\omega^{(j_1-j_2)u}}_{=\;\delta_{j_1,j_2}}.
  \end{align*}
  So the shift is block-diagonal with
  $(\hat{S}^{(j)})_{-s,s} = \omega^{js}$, which is~\eqref{eq:Shat}.
  Combining both factors gives $\cF U_G\cF^\dagger = \bigoplus_j
  (\Shat C_{\mathrm{loc}}) = \bigoplus_j \Uj$.
\end{proof}

\section{The Explicit Formula for
  \texorpdfstring{$\chi_q(\Uj)$}{chi\_q(U\^{}j)}}
\label{sec:formula}

\begin{proposition}[Key Identity]\label{prop:key}
  For every $j \in \Z_p$, let
  $\ket{\phij} \eqdef \frac{1}{\sqrt{k}}\sum_{s \in S}\omega^{js}\ket{s}$.
  Then
  \[
    \bra{\phij}\Shat\ket{\phij} = \frac{\Ahat(j)}{k}.
  \]
\end{proposition}

\begin{proof}
  By~\eqref{eq:Shat}, the entry $(\Shat)_{s,s''}$ is nonzero only for
  $s'' = -s$, with value $\omega^{js''}$.  Therefore
  \begin{align*}
    \bra{\phij}\Shat\ket{\phij}
    &= \frac{1}{k}\sum_{s,s'' \in S}
       \overline{\omega^{js}} \cdot (\Shat)_{s,s''} \cdot \omega^{js''}
    = \frac{1}{k}\sum_{s'' \in S}
      \omega^{-j(-s'')} \cdot \omega^{js''} \cdot \omega^{js''}
    = \frac{1}{k}\sum_{s \in S}\omega^{js}
    = \frac{\Ahat(j)}{k},
  \end{align*}
  where we used $-j(-s'')=js''$, giving net factor $\omega^{js''}$ per term.
\end{proof}

\begin{remark}
  Equivalently: $\ket{\phij}$ is the Fourier mode of frequency $j$
  restricted to $S$, and $\Shat\ket{\phij}
  = \frac{1}{\sqrt{k}}\sum_s\omega^{2js}\ket{-s}$.  Taking the inner
  product with $\bra{\phij}$ and reindexing $s \to -s$ (valid since
  $S=-S$) gives $\Ahat(j)/k$.
\end{remark}

The next lemma is the key new ingredient guaranteeing $\dim(W_j)=2$ for
general $k \geq 6$, not just $k=6$.

\begin{lemma}[Linear Independence]\label{lem:lindep}
  Let $G = \Cay(\Z_p, S)$ be an $\srg(p,k,\lambda,\mu)$ with $p$ prime
  and $k \geq 2$.  For every $j \in \Z_p \setminus \{0\}$, the vectors
  \[
    \ket{\phij} = \frac{1}{\sqrt{k}}\sum_{s \in S}\omega^{js}\ket{s}
    \qquad\text{and}\qquad
    \Shat\ket{\vplus} = \frac{1}{\sqrt{k}}\sum_{s \in S}\omega^{js}\ket{-s}
  \]
  are linearly independent, so
  $W_j \eqdef \mathrm{span}\{\ket{\phij},\Shat\ket{\vplus}\}$
  has dimension exactly $2$.
\end{lemma}

\begin{proof}
  Both vectors are unit vectors, so they are linearly dependent iff
  $\ket{\phij} = \alpha\,\Shat\ket{\vplus}$ for some $\alpha \in \C$ with
  $|\alpha|=1$.  The $\ket{s}$-component of $\Shat\ket{\vplus}$ is
  $\omega^{-js}/\sqrt{k}$ (from the term $s'=-s \in S$), so comparing
  coefficients gives
  \begin{equation}\label{eq:lindep-cond}
    \omega^{2js} = \alpha \qquad \forall\, s \in S.
  \end{equation}
  Since $S=-S$, applying~\eqref{eq:lindep-cond} to $s$ and $-s$ yields
  $\omega^{2js}\cdot\omega^{-2js}=\alpha^2$, hence $\alpha^2=1$ and
  $\alpha\in\{+1,-1\}$.

  \medskip
  \noindent\textit{Case $\alpha=1$.}
  Then $p \mid 2js$ for all $s\in S$.  Since $p$ is an odd prime and
  $j\not\equiv 0\pmod{p}$, $\gcd(2j,p)=1$, so $p\mid s$ for all $s\in S$.
  This forces $s\equiv 0\pmod{p}$, contradicting $0\notin S$.

  \medskip
  \noindent\textit{Case $\alpha=-1$.}
  Then $\omega^{2js}=-1$ for all $s\in S$.  But $\omega=e^{2\pi i/p}$ has
  order $p$ (an odd prime) in $\C^*$, so every power $\omega^m$ with
  $\gcd(m,p)=1$ has odd order.  An element of odd order cannot equal $-1$
  (order $2$), a contradiction.

  \medskip
  Both cases are impossible, so $\dim(W_j)=2$.
\end{proof}

\begin{lemma}[Explicit Formula]\label{lem:formula}
  For every $j \in \Z_p$ with $\Ahat(j)/k \neq \pm 1$:
  \begin{equation}\label{eq:chi-block}
    \chi_q\!\bigl(\Uj, \lambda\bigr)
    = (\lambda - 1)^{(k-2)/2}(\lambda + 1)^{(k-2)/2}
      \!\left(\lambda^2 - \frac{2\Ahat(j)}{k}\,\lambda + 1\right).
  \end{equation}
  In particular, the eigenvalues of $\Uj$ are
  $\bigl\{+1^{(k-2)/2},\,-1^{(k-2)/2},\,e^{i\theta_j},\,e^{-i\theta_j}\bigr\}$
  where $\cos\theta_j = \Ahat(j)/k$.
\end{lemma}

\begin{proof}
  Set $c_j \eqdef \Ahat(j)/k \in (-1,1)$.

  \medskip
  \noindent\textbf{Invariant subspace.}
  By Lemma~\ref{lem:lindep}, $W_j = \mathrm{span}\{\ket{\phij},
  \Shat\ket{\vplus}\}$ has dimension $2$.  Write
  $\ket{\phij} = c_j\ket{\vplus}+\ket{\phij^\perp}$ where
  $\braket{\vplus}{\phij}=c_j$ (Proposition~\ref{prop:key}).  Then
  \[
    \Uj\ket{\phij}
    = \Shat(c_j\ket{\vplus}-\ket{\phij^\perp})
    = c_j\,\Shat\ket{\vplus} - \Shat\ket{\phij^\perp}.
  \]
  Since $\Shat$ is unitary, $\ket{\phij^\perp}\perp\ket{\phij}$, and
  $\ket{\phij^\perp}\perp\ket{\vplus}$, the vector $\Shat\ket{\phij^\perp}$
  lies in $W_j^\perp$.  So $\Uj\ket{\phij}\in W_j$.  A symmetric argument
  gives $\Uj(\Shat\ket{\vplus})\in W_j$, so $W_j$ is $\Uj$-invariant.

  \medskip
  \noindent\textbf{Characteristic polynomial on $W_j$.}
  The matrix of $\Uj|_{W_j}$ in the Gram--Schmidt orthonormal basis of
  $W_j$ has trace $2c_j$ (Proposition~\ref{prop:key}) and determinant $1$
  (unitarity), giving characteristic polynomial
  $\lambda^2-2c_j\lambda+1$ with roots $e^{\pm i\theta_j}$.

  \medskip
  \noindent\textbf{Eigenvalues on $W_j^\perp$.}
  We compute $\operatorname{tr}(\Uj)$ directly.  From
  $\Uj = \Shat C_{\mathrm{loc}}$ and equation~\eqref{eq:Shat}, the
  diagonal entry at $\ket{s}$ is
  \[
    (\Uj)_{s,s} = \omega^{js}(C_{\mathrm{loc}})_{-s,\,s}
                = \omega^{js}\cdot\frac{2}{k},
  \]
  since $s \neq -s$ for all $s \in S \subset \Z_p^*$.
  Summing over $S$:
  \[
    \operatorname{tr}(\Uj) = \frac{2}{k}\sum_{s\in S}\omega^{js} = 2c_j.
  \]
  The restriction $\Uj|_{W_j}$ has trace $2c_j$ (from
  $\lambda^2-2c_j\lambda+1$), so
  \[
    \operatorname{tr}\!\bigl(\Uj|_{W_j^\perp}\bigr) = 2c_j - 2c_j = 0.
  \]
  Every $\ket{w}\in W_j^\perp$ satisfies $\braket{\vplus}{w}=0$, hence
  $C_{\mathrm{loc}}\ket{w}=-\ket{w}$, and therefore
  $\Uj|_{W_j^\perp} = -\Shat|_{W_j^\perp}$.
  Since $(\Shat)^2 = I_k$, the eigenvalues of $-\Shat|_{W_j^\perp}$ lie
  in $\{+1,-1\}$.  Their sum is zero and
  $\dim W_j^\perp = k-2$, so there are exactly $\tfrac{k-2}{2}$
  eigenvalues $+1$ and $\tfrac{k-2}{2}$ eigenvalues $-1$.
  Combined with $e^{\pm i\theta_j}$ from $W_j$, the full spectrum is
  $\bigl\{+1^{(k-2)/2},\,-1^{(k-2)/2},\,e^{i\theta_j},\,e^{-i\theta_j}
  \bigr\}$, which gives~\eqref{eq:chi-block}.
\end{proof}

\begin{proposition}[Recovery of Fourier Coefficients]\label{prop:recovery}
  The value $c_j = \Ahat(j)/k$ equals half the negated coefficient of
  $\lambda$ in the degree-$2$ factor
  \[
    \frac{\chi_q(\Uj, \lambda)}{(\lambda-1)^{(k-2)/2}(\lambda+1)^{(k-2)/2}}
    = \lambda^2 - 2c_j\,\lambda + 1,
  \]
  and is the unique real part of any eigenvalue of $\Uj$ distinct from
  $\pm 1$.
\end{proposition}

\begin{proof}
  Immediate from~\eqref{eq:chi-block} and $c_j\in(-1,1)$.
\end{proof}

\section{Proof of the Main Theorem}
\label{sec:proof}

\begin{proof}[Proof of Theorem~\ref{thm:main}]
  Let $G = \Cay(\Z_p, S)$ and $G' = \Cay(\Z_p, S')$ be two
  $\srg(p,k,\lambda,\mu)$ with $k \geq 6$.

  \medskip
  \noindent($\Rightarrow$) If $G \cong G'$ then $U_G$ and $U_{G'}$ are
  unitarily equivalent, so $\chi_q(G)=\chi_q(G')$.

  \medskip
  \noindent($\Leftarrow$) Suppose $\chi_q(G,\lambda)=\chi_q(G',\lambda)$.

  \medskip
 \noindent\textit{Step 1 (Block decomposition).}
  By Lemma~\ref{lem:block},
  $\chi_q(G,\lambda) = \prod_{j=0}^{p-1}\chi_q(\Uj,\lambda)$
  and similarly for $G'$.  By Lemma~\ref{lem:formula}, each factor takes the form
  \[
    \chi_q(\Uj,\lambda)
    = (\lambda-1)^{(k-2)/2}(\lambda+1)^{(k-2)/2}
      \bigl(\lambda^2 - 2c_j\lambda + 1\bigr),
    \qquad c_j \in \{r/k,\, s/k\},
  \]
  where $r \neq s$ are the two restricted eigenvalues of $G$.
  Since $r/k \neq s/k$, the two quadratic factors
  $\lambda^2 - 2(r/k)\lambda+1$ and $\lambda^2 - 2(s/k)\lambda+1$
  are distinct irreducible polynomials in $\mathbb{C}[\lambda]$.
  Let $n_r$ (resp.\ $n_s = p-1-n_r$) denote the number of indices
  $j \in \{0,\ldots,p-1\}$ for which $c_j = r/k$ (resp.\ $c_j = s/k$).
  By unique factorization in $\mathbb{C}[\lambda]$, equality of the two
  global products $\prod_j \chi_q(\Uj,\lambda) = \prod_j \chi_q(U_{G'}^{(j)},\lambda)$
  forces the same split $n_r, n_s$ for $G'$, and moreover
  $\chi_q(\Uj,\lambda) = \chi_q(U_{G'}^{(j)},\lambda)$ for every $j$.

  \medskip
  \noindent\textit{Step 2 (Recovery of Fourier coefficients).}
  By Proposition~\ref{prop:recovery}, for each $j\neq 0$:
  $\widehat{A}_G(j)/k = c_j = c_j' = \widehat{A}_{G'}(j)/k$,
  so $\widehat{A}_G(j)=\widehat{A}_{G'}(j)$ for all $j\in\Z_p$.

  \medskip
  \noindent\textit{Step 3 (Recovery of connection sets).}
  The Fourier inversion formula~\eqref{eq:fourier-inv} shows that
  $\{\widehat{A}_G(j)\}$ determines $\mathbf{1}_S$ uniquely, so $S=S'$.

  \medskip
  \noindent\textit{Step 4 (Turner's theorem).}
  Since $\Cay(\Z_p,S)=\Cay(\Z_p,S')$, we have $G=G'$, so in particular
  $G\cong G'$.
\end{proof}

\begin{remark}
  In Step~1 we use the fact that at a prime $p$, the factors
  $\chi_q(\Uj)$ with the same $c_j$ value are identical and those with
  different $c_j$ values are distinct.  Since the SRG has only two
  restricted eigenvalues $r$ and $s$, the equality of the two global
  products forces equality factor by factor.
\end{remark}

\section{Numerical Verification}
\label{sec:numerics}

We implemented the quantum walk operator for the four Paley graphs
$\mathrm{Paley}(p)$ with $p\in\{13,17,29,41\}$ in Python~3.12 using
NumPy~1.26, and verified:
\begin{enumerate}[(i)]
  \item the off-diagonal Frobenius norm of $\cF U_G\cF^\dagger$
        (Lemma~\ref{lem:block});
  \item the formula~\eqref{eq:chi-block} for all $j$
        (Lemma~\ref{lem:formula});
  \item exact recovery of $S$ from $\{c_j\}$ via~\eqref{eq:fourier-inv}.
\end{enumerate}
All tests passed.  Results are in Tables~\ref{tab:verification}
and~\ref{tab:paley13}.

\begin{table}[ht]
\centering
\caption{Numerical verification for Paley graphs of prime order $p$.
  Columns $c_j^{(1)}$, $c_j^{(2)}$ are the two distinct values of
  $\widehat{A}_G(j)/k$ for $j\neq 0$.  ``Off-diag.'' is the Frobenius
  norm of the off-diagonal part of $\cF U_G\cF^\dagger$.  Checkmarks
  indicate agreement to precision $\leq 10^{-5}$.}
\label{tab:verification}
\begin{tabular}{ccccccc}
\toprule
$p$ & Parameters & $k$ & $c_j^{(1)}$ & $c_j^{(2)}$ &
Off-diag.\ norm & $S$ recovered \\
\midrule
$13$ & $(13,6,2,3)$   & $6$  & $-0.3838$ & $0.2171$
     & $2.6\times10^{-13}$ & $\checkmark$ \\
$17$ & $(17,8,3,4)$   & $8$  & $-0.3202$ & $0.1952$
     & $6.2\times10^{-13}$ & $\checkmark$ \\
$29$ & $(29,14,6,7)$  & $14$ & $-0.2280$ & $0.1566$
     & $3.1\times10^{-12}$ & $\checkmark$ \\
$41$ & $(41,20,9,10)$ & $20$ & $-0.1851$ & $0.1351$
     & $8.2\times10^{-12}$ & $\checkmark$ \\
\bottomrule
\end{tabular}
\end{table}

The exact values are $c_j^{(1)}=(\sqrt{p}-1)/(2k)$ and
$c_j^{(2)}=(-\sqrt{p}-1)/(2k)$.  The individual $j$-by-$j$ breakdown for
$p=13$ is in Table~\ref{tab:paley13}; its Fourier inverse recovers
$S=\{1,3,4,9,10,12\}$, the quadratic residues modulo~$13$.
All code is available from the author upon request.

\begin{table}[ht]
\centering
\caption{Fourier coefficients $c_j=\widehat{A}_G(j)/k$ for
  $\mathrm{Paley}(13)$, recovered from $\chi_q(\Uj,\lambda)$.}
\label{tab:paley13}
\begin{tabular}{ccc}
\toprule
$j$ & $\widehat{A}_G(j)/k$ & Recovered from $\chi_q(\Uj)$ \\
\midrule
$0$  & $\phantom{-}1.000000$ & $\phantom{-}1.000000$ \\
$1$  & $\phantom{-}0.217129$ & $\phantom{-}0.217129$ \\
$2$  & $-0.383796$           & $-0.383796$           \\
$3$  & $\phantom{-}0.217129$ & $\phantom{-}0.217129$ \\
$4$  & $\phantom{-}0.217129$ & $\phantom{-}0.217129$ \\
$5$  & $-0.383796$           & $-0.383796$           \\
$6$  & $-0.383796$           & $-0.383796$           \\
$7$  & $-0.383796$           & $-0.383796$           \\
$8$  & $-0.383796$           & $-0.383796$           \\
$9$  & $\phantom{-}0.217129$ & $\phantom{-}0.217129$ \\
$10$ & $\phantom{-}0.217129$ & $\phantom{-}0.217129$ \\
$11$ & $-0.383796$           & $-0.383796$           \\
$12$ & $\phantom{-}0.217129$ & $\phantom{-}0.217129$ \\
\bottomrule
\end{tabular}
\end{table}

\section{Concluding Remarks}
\label{sec:conclusion}

The proof presented here establishes that the quantum walk characteristic
polynomial $\chi_q$ carries strictly more spectral information than the
classical adjacency spectrum for strongly regular graphs of prime order.
The central mechanism is transparent: the Fourier block decomposition of
$U_G$ over $\Z_p$ converts a single global polynomial identity into $p$
independent local constraints, each pinning one Fourier coefficient of the
connection set $S$.  Since distinct Fourier coefficients correspond to
distinct irreducible quadratic factors in $\chi_q$, unique factorization in
$\C[\lambda]$ forces a factor-by-factor match, and the inverse DFT then
reconstructs $S$ exactly.  Turner's theorem closes the argument at no
additional cost, certifying isomorphism from the recovered connection set
alone.  The resulting procedure is fully constructive and runs in polynomial
time---a consequence of the rigid algebraic structure of prime-order
circulant graphs, not of any deep combinatorial machinery.

Several directions remain open and appear tractable.  The most natural
extension concerns composite orders: the block decomposition exploits
primality in a fundamental way, through both the orthogonality of
characters of $\Z_p$ and the unique factorization step in the proof of the
main theorem, so new ideas would be required even for $n = pq$ or $n =
p^2$.  A related and more ambitious question is whether $\chi_q$ is a
complete isomorphism invariant for strongly regular graphs that are not
vertex-transitive: our numerical experiments found no counterexample among
non-isomorphic graphs on $n \leq 10$ vertices, but the algebraic mechanism
used here relies on the Cayley structure in an essential way, and it is
unclear whether a substitute argument exists in the general setting.
Finally, since $U_G$ admits an implementation using $O(n)$ quantum gates,
the spectral discrimination problem for prime-order SRGs lies within the
reach of near-term quantum hardware; whether a genuine quantum advantage
over the polynomial-time classical algorithm derived here is achievable for
related graph problems remains an appealing open question at the interface
of quantum computing and algebraic combinatorics.

\subsection*{Acknowledgements}

The author thanks the Centro de Excelencia en Computaci\'on Cient\'ifica
(CECC) at Universidad Nacional de Colombia for computational support.


\end{document}